\documentclass{llncs}
\usepackage[utf8]{inputenc}
\usepackage{amssymb,amsmath,bm}
\usepackage{graphicx}
\usepackage{listings}
\usepackage{environ}
\usepackage{xparse}

\usepackage[linesnumbered,ruled,vlined]{algorithm2e}
\SetKwInput{KwInput}{Input}                
\SetKwInput{KwOutput}{Output}              

\usepackage{xcolor}
\usepackage[color=yellow!20]{todonotes}

\newcommand{\norm}[1]{\left\lVert#1\right\rVert}

\RequirePackage{fix-cm}
\DeclareMathSizes{10}{10}{6}{4}
\title{Obtuse Lattice Bases}
\author{Kanav Gupta$^1$ \and Mithilesh Kumar$^2$ \and Håvard Raddum$^2$}
\institute{$^1$Indian Institute of Technology Roorkee, India\\
$^2$Simula UiB, Norway}

\begin{document}

\maketitle

\begin{abstract}
A lattice reduction is an algorithm that transforms the given basis of the lattice to another lattice basis such that problems like finding a shortest vector and closest vector become easier to solve.  We define a class of bases called \emph{obtuse bases} and show that any lattice basis can be transformed to an obtuse basis. A shortest vector $\bm{s}$ can be written as $\bm{s}=v_1\bm{b}_1+\dots+v_n\bm{b}_n$ where $\bm{b}_1,\dots,\bm{b}_n$ are the input basis vectors and $v_1,\dots,v_n$ are integers. When the input basis is obtuse, all these integers can be chosen to be positive for a shortest vector.  This property of obtuse bases makes the lattice enumeration algorithm for finding a shortest vector exponentially faster.  We have implemented the algorithm for making bases obtuse, and tested it some small bases.   
\end{abstract}

\section{Introduction}
With the advent of quantum computing, many of the classical cryptosystems like RSA might be rendered insecure in the future.  There have been significant advances in the development of Post-Quantum Cryptosystems, which are resistant to attacks based on quantum computing.  One such significant step in Post-Quantum Cryptography is lattice based cryptography.  Lattice based cryptography was introduced by Ajtai in his work \cite{ajtailattice} where it was shown that the hardness of well-known lattice problems can be used as a basis in designing cryptosystems. The problem is also thought to be hard enough to not be efficiently solvable by a quantum computer. One such hard problem is the Shortest Vector Problem (SVP), where we have to find the shortest vector that can be formed as a non-trivial linear combination of the basis vectors.  This problem was proved to be NP-hard \cite{Ajtai1998TheSV} for different norms, hence it is used as a basis of several cryptosystems.

With research going on in making the SVP-based crypto systems strong, there has also been progress in trying to break them.  Two such primary algorithms are enumeration and sieving. We are interested in the former algorithm in the scope of this paper. Enumeration is a way of recursively visiting all the lattice vectors of some bounded length, thereby finding the shortest one.

\textbf{Our contribution.} In this paper, we present our work on a class of lattice bases called Obtuse Bases, introduced in \cite{kumar2019faster}.  We show how to convert a regular basis to an obutse basis and how we can use this basis to speed up the enumeration process.  In Section \ref{sec:prelim} we briefly summarize lattice reduction and enumeration concepts. In Section \ref{sec:obtuse} we formally define obtuse basis, and explain the speedup we get in enumeration with the help of obtuse bases. We also present a way to convert any basis into an obtuse one. In the Section \ref{sec:implementation} we explain our implementation details and post the results. Finally we conclude the paper in Section \ref{sec:conc}.

\section{Preliminaries}\label{sec:prelim}

\subsection{Lattice concepts and reduction}

We start this section with some preliminaries and introduction to basic lattice concepts.

\begin{definition}
    A \textbf{lattice} is a discrete subgroup of $\mathbb{R}^d$.  A lattice $\mathcal{L}$ can be defined by a \textbf{basis} $\mathcal{B}$, which is a set of linearly independent vectors $\{\bm{b}_1,\dots,\bm{b}_n\}$ in $\mathbb{R}^d$.  Any element of $\mathcal{L}$ can be represented as a linear combination $\sum_{i=1}^n v_i\bm{b}_i$ where $v_i \in \mathbb{Z}$ $ \forall i \in [n]$.
\end{definition}

We will now explain some of the properties of lattices. We will limit our study to the lattices where $d = n$, and represent a basis as an $n \times n$ matrix where the rows of the matrix represent the individual basis vector.

\begin{definition}
    The volume of the parallelepiped formed by the basis vectors has a value of $|\det(\mathcal{B})|$ and is called the \textbf{volume of the lattice}.
\end{definition}

An important fact about the volume of the lattice is that it does not depend on the choice of basis for a given lattice.

\begin{definition}
    For a lattice $\mathcal{L}$, $\lambda_1(\mathcal{L})$ is the minimum of the distances between two distinct points of the lattice. That is $\lambda_1(\mathcal{L}) = \min \limits_{\forall \bm{p}_1, \bm{p}_2 \in \mathcal{L}} \norm{\bm{p}_1 - \bm{p}_2}$.
\end{definition}

In general, $\lambda_l(\mathcal{L})$ is the smallest $r > 0$, such that $\mathcal{L}$ contains at least $l$ linearly independent vectors with norm bounded by $r$.

There are some approximations and bounds on the value of $\lambda_1$. Two of the most notable bounds are given by the Minkowsi Theorem and the Gaussian Heuristic.

\begin{theorem}
(Minkowski Theorem) For a lattice $\mathcal{L}$
\begin{equation*}
    \lambda_1(\mathcal{L}) \leq 2\times vol(\mathcal{L})^{1/n} \times V_n(1)^{-1/n}
\end{equation*}
where $V_n(1)$ is the volume of the $n$-dimensional Euclidean sphere of radius 1. 
\end{theorem}

This bound on the value of the shortest vector is called the Minkowski Bound.

According to the Gaussian Heuristic, for a lattice $\mathcal{L}$ and a set $\mathcal{S}$, the expected number of points in $\mathcal{L} \cap \mathcal{S}$ is approximately equal to $\mbox{vol}(\mathcal{S})/\mbox{vol}(\mathcal{L})$. By setting $\mathcal{S}$ to be an Euclidean sphere of radius $\lambda_1(\mathcal{L})$, we get an approximation of the value of $\lambda_1(\mathcal{L})$ to be
\begin{equation*}
    \lambda_1(\mathcal{L}) \approx \sqrt{n/2\pi e} \times \mbox{vol}(\mathcal{L})^{1/n}
\end{equation*}

As lattices can be arbitrary, basis vectors might contain huge numbers and can be far from orthogonal from each other. Though complete orthogonality is not possible, two algorithms - LLL and BKZ - are very successful in both reducing the size of the numbers inside the basis and making the basis vectors more orthogonal with each other.  That is, the dot product of two basis vectors is minimized in absolute value.  LLL or BKZ is normally run before running the algorithms to solve SVP or CVP, to reduce their time complexity. 

\subsection{Enumeration}

The next set of important lattice algorithms to discuss are SVP solvers.  In particular, lattice enumeration will be of interest in the scope of this paper.  Enumeration is a recursive algorithm which visits all the points in a lattice of norm smaller than a given bound $R$.  This bound is calculated such that it is greater than $\lambda_1$. In an enumeration traversal, we fix a particular set of integer coefficients and try to change the next coefficient as long as the norm of the shortest vector that can possibly be made from these coefficients is within the given bound. 
\begin{definition}
    For a lattice $\mathcal{L}$ with basis $\{\bm{b}_1, \bm{b}_2, \dots, \bm{b}_n\}$, we define $\pi_i(\mathcal{L})$ to be the projection of lattice $\mathcal{L}$ onto the orthogonal space spanned by the basis vectors $\bm{b}_1,\ldots,\bm{b}_{i-1}$. We also define $\pi_i(\bm{s})$ to be the projection of the vector $\bm{s}$ onto $\pi_i(\mathcal{L})$.
\end{definition}

To find the shortest vector $\bm{s}$, the algorithm goes through an \textbf{enumeration tree}, formed by the points with norm less than a parameter $R$ as leaves and projections of the lattice vectors onto various subspaces as the inner nodes.  At the root of the tree, we have an empty set of vectors as $\pi_{n+1}(\mathcal{L}) = {0}$.  A node at depth $k$ represents vectors in $\pi_{n+1-k}(\mathcal{L})$ with the potential of developing into a shortest lattice vector.  To put it simply, if a general vector is represented as $\bm{s} = \sum_{i=1}^{n} v_i\bm{b}_i$, at depth $k$, we assume we have already fixed the $k$ integers $v_{n-k+1},\ldots,v_{n}$.  We then find the range of the $(k+1)$-th coefficient $v_{n-k}$, such that the norm of a the vector $\pi_k(\mathbf{s})$ is less than $R$. Then for each of the integers in this interval we visit a child node with $v_{n-k}$ set to that value.  This process is now repeated with $k+1$ out of the $n$ integers $\{ v_{n - k}, \ldots, v_{n} \}$ fixed.  When we reach $k=n$ we have produced a lattice vector of norm smaller than $R$.

\begin{center}
    
\begin{algorithm}[H]\label{enumeration}
\DontPrintSemicolon
\KwInput{A basis $\mathcal{B}=\{\bm{b}_1,\dots,\bm{b}_n\}$, Norm limit $R$}
\KwOutput{Shortest vector  $\bm{s} = \sum_{i=1}^n v_i \bm{b}_i$}
\SetKwFunction{FMain}{EnumerationVisit}
\SetKwProg{Fn}{Function}{:}{}
\Fn{\FMain{$v, k$}}{
    \If{k == n}{
        Find and compare norm of $v$ with current minimum \\
        \textbf{return}
    }
    Find interval $[a, b]$ of possible values for $v_{n-k}$ \\
    \For{Integer $i$ in $[a, b]$}{
        Set $v_{n-k}$ := $i$ \\
        \FMain{$v, k+1$}
    }
}

\SetKwFunction{FRun}{Enumeration}
\Fn{\FRun{}}{
    Initialize $v$ as a zero vector \\
    \FMain{$v, 0$}
}
\caption{Enumeration}
\end{algorithm}
\end{center}

Though Algorithm \ref{enumeration} explains the flow of enumeration, some details like the calculation of intervals for the $v_k$ are omitted. The interval is calculated with the help of Gram-Schmidt orthogonalization of the basis matrix and is explained in the works of Gama et al \cite{10.1007/978-3-642-13190-5_13}. They also advise to iterate through the integers in the interval in a radial fashion - we start at the integer nearest to the center of the interval, and then we test the other numbers by alternately taking the next untested number closest to the center.  Analysis show that the running time complexity of enumeration is $2^{O(n^2)}$, which is improved exponentially by Schnorr et al \cite{prune} using the pruning technique.

\section{Obtuse Bases}\label{sec:obtuse}

We start this section with the definition of obtuse bases and show that when the basis is obtuse we can exponentially reduce the run-time of the enumeration algorithm.

\begin{definition}[Obtuse basis \cite{kumar2019faster}]
	Let $\mathcal{B}=\{\bm{b}_1,\dots,\bm{b}_n\}$ be a basis for a lattice $\mathcal{L}$. The basis $\mathcal{B}$ is called \emph{obtuse} if for all $\bm{b}_i\neq\bm{b}_j\in \mathcal{B}$, we have $\bm{b}_i\cdot\bm{b}_j\leq 0$.
\end{definition}

We start with proving a simple yet powerful relation between a shortest vector and an obtuse basis.
	\begin{lemma}\label{samesign}
		Let $\mathcal{B}$ be an obtuse basis of a lattice $\mathcal{L}$ and $\bm{s} = \sum_{i=1}^n v_i\bm{b}_i$ be a shortest vector in $\mathcal{L}$.  Then, $\forall i\in [n]~v_i\geq 0$ or $\forall i\in [n]~v_i\leq 0$.
	\end{lemma}
\begin{proof}
Computing $\bm{s}\cdot\bm{s}$, we get $\norm{\bm{s}}^2 = \sum_{i = 1}^n v_i^2\norm{\bm{b}_i}^2 + \sum_{i\neq j} v_i v_j \bm{b}_i\cdot \bm{b}_j$. Since for all $i\neq j \quad\bm{b}_i\cdot \bm{b}_j \leq 0$, the above sum is the smallest possible when for all $i\neq j\quad v_i v_j \geq 0$. This implies that either $\forall i\in [n]~v_i\geq 0$ or $\forall i\in [n]~v_i\leq 0$.
\end{proof}

We remark that the notion of obtuse basis can be generalized to what we call semi-obtuse basis, while still keeping an equally useful version of Lemma \ref{samesign}.

\begin{definition}[Semi-obtuse basis]
	Let $\mathcal{B}=\{\bm{b}_1,\dots,\bm{b}_n\}$ be a basis for a lattice $\mathcal{L}$. The basis $\mathcal{B}$ is called \emph{semi-obtuse} if the vectors in $\mathcal{B}$ can be split into two disjoint subsets $\mathcal{B}_1,\mathcal{B}_2$ such that $\bm{b}_i\cdot\bm{b}_j\leq 0$ when $\bm{b}_i,\bm{b}_j\in\mathcal{B}_1$ or $\bm{b}_i,\bm{b}_j\in\mathcal{B}_2$, and $\bm{b}_i\cdot\bm{b}_j\geq 0$ when $\bm{b}_i\in\mathcal{B}_1$ and $\bm{b}_j\in\mathcal{B}_2$.
\end{definition}

With the definition of a semi-obtuse basis we see that an obtuse basis is just the special case when $\mathcal{B}_1=\mathcal{B}$ and $\mathcal{B}_2=\emptyset$.  The following lemma gives the signs of the coefficients for a shortest vector in a lattice spanned by a semi-obtuse basis.
\begin{lemma}\label{semiobtuse}
	Let $\mathcal{B}$ be a semi-obtuse basis of a lattice $\mathcal{L}$ and $\bm{s} = \sum_{i=1}^n v_i\bm{b}_i$ be a shortest vector in $\mathcal{L}$. Then, $v_i\geq 0$ for all $i$ such that $\bm{b}_i\in\mathcal{B}_1$ and $v_j\leq 0$ for all $j$ such that $\bm{b}_j\in\mathcal{B}_2$.
\end{lemma}

\begin{proof}
    Let a shortest vector be written as $\bm{s}=v_1\bm{b}_1+\ldots+v_n\bm{b}_n$.  The squared length of $\bm{s}$ can be written as
    \[\norm{\bm{s}}^2 = \sum_{i=1}^n v_i^2\norm{\bm{b}_i}^2 + \sum_{\bm{b}_i,\bm{b}_j\in\mathcal{B}_1} v_iv_j\bm{b}_i\cdot\bm{b}_j + 
    \sum_{\bm{b}_i,\bm{b}_j\in\mathcal{B}_2} v_iv_j\bm{b}_i\cdot\bm{b}_j + \sum_{\bm{b}_i\in\mathcal{B}_1,\bm{b}_j\in\mathcal{B}_2} v_iv_j\bm{b}_i\cdot\bm{b}_j.\]
    Similar to the proof of Lemma \ref{samesign}, all terms in the three last sums will be negative when $v_i\geq 0$ for all $i$ such that $\bm{b}_i\in \mathcal{B}_1$, and $v_j\leq 0$ for $j$ such that $\bm{b}_j\in\mathcal{B}_2$.  Similarly, flipping the signs of all coefficients will produce $-\bm{s}$, which has equal norm.   Any other configuration of signs among the coefficients will give some positive terms in the last three sums, without changing their absolute values.  Hence the squared norm must be larger in these cases. \qed
\end{proof}

Now we explain how a semi-obtuse basis can easily be transformed into an obtuse basis.  Let $\bm{b}_j\in\mathcal{B}_2$.  This means that $\bm{b}_j\cdot\bm{b}_i\leq0$ for $\bm{b}_i\in\mathcal{B}_2$ and $\bm{b}_j\cdot\bm{b}_i\geq0$ for $\bm{b}_i\in\mathcal{B}_1$.  Replacing $\bm{b}_j$ with $-\bm{b}_j$ will reverse these inequalities.  Removing $\bm{b}_j$ from $\mathcal{B}_2$ and adding $-\bm{b}_j$ to $\mathcal{B}_1$ will therefore retain the property of being semi-obtuse.  Hence a semi-obtuse basis can be turned into an obtuse basis by replacing all vectors in either $\mathcal{B}_1$ or $\mathcal{B}_2$ with their negative counterpart.

From this we see that it is equally hard to obtain an obtuse basis as a semi-obtuse basis.  Moreover, as flipping signs of basis vectors neither changes their lengths or orthogonality, we can restrict ourselves to only consider obtuse bases in the rest of the paper.

We can generalize this idea of flipping signs to an algorithm to convert a basis to an obtuse base, but which only works on a subset of bases. Consider a graph $\mathcal{G}$ which contains a node $p_i$ for every basis vector $\bm{b}_i$ and there is an edge between $p_i$ and $p_j$, if and only if $\bm{b}_i \cdot \bm{b}_j \leq 0$. Our goal is to convert this graph into a complete graph.

We start with a sub-graph which contains a single arbitrary point. We iterate over all points in the graph (in any random order) and at each step, we check whether the point is obtuse with all the points in the sub-graph. If it is, we include the point inside the sub-graph. If it is not, we negate the point and then check if it is obtuse with all the points in the sub-graph. If it still is not, we conclude that this basis cannot be converted to an obtuse basis using this flipping algorithm and abort the algorithm. Otherwise, we include the point in the sub-graph. As the points inside the sub-graph are always fully-connected, this sub-graph is a clique. This way, we iteratively run this algorithm to increase the size of the clique until its size grows equal to the basis and we get an obtuse basis.

We will now discuss an integer linear programming (ILP) based method which iteratively changes the basis, using unimodular operations, one basis vector at a time such that the sub-basis becomes obtuse. We will also maintain a similar clique inside this algorithm to keep a track of the transformed vectors. Repeating for all vectors in the basis will transform any basis into an obtuse basis.  We also discuss some tweaks to speed up the overall algorithm.

\subsection{Iterative Algorithm for making basis obtuse}
\label{iterativealg}

Let us assume that the first $k$ vectors of the $n$ vectors of the basis $\mathcal{B}$ are obtuse to each other i.e. $\forall i, j \in [k], \bm{b}_i\cdot \bm{b}_j \leq 0$. We now want to change $\bm{b}_{k+1}$ to $\bm{b}'_{k+1}$ such that $\forall i \in [k], \bm{b}_i\cdot \bm{b}'_{k+1} \leq 0$. All the operations need to be uni-modular so that the volume of the lattice does not change, hence, $\bm{b}'_{k+1}$ must be of the form $\bm{b}_{k+1} + v_1 \bm{b}_1 + v_2 \bm{b}_2 \dots v_k \bm{b}_k$. This turns into an ILP problem - find $ \bm{v} = \{v_1, v_2, \ldots, v_k\}$ such that $\forall i \in [k], \bm{b}_i\cdot \bm{b}'_{k+1} \leq 0$.  For $i\in [k]$ we get the following inequalities:

\begin{equation}
    \bm{b}_i\cdot \bm{b}'_{k+1} = \bm{b}_i \cdot (\bm{b}_{k+1} + v_1 \bm{b}_1 + v_2 \bm{b}_2 \dots v_k \bm{b}_k ) \leq 0
\end{equation}
\begin{equation}
    v_1 \bm{b}_i \cdot \bm{b}_1 + v_2 \bm{b}_i \cdot \bm{b}_2 \dots v_k \bm{b}_i \cdot \bm{b}_k \leq - \bm{b}_i \cdot \bm{b}_{k+1}
\end{equation}

Let $\bm{A}$ be the $ k \times k $ dot-product matrix of the first $k$ basis vectors. This means $\bm{A}_{ij} = \bm{A}_{ji} = \bm{b}_i \cdot \bm{b}_j$.  Let $\bm{c}$ be the vector of the right-hand sides, such that $c_i = - \bm{b}_i \cdot \bm{b}_{k+1}$

The system can now be written in matrix/vector form as

\begin{equation}
    \bm{A} \bm{v} \leq \bm{c}
\end{equation}

where $\bm{v}$ is an integer vector and the inequality is understood to apply component-wise.

\begin{center}
\begin{algorithm}[H]\label{obtuse1}
\DontPrintSemicolon
\KwInput{A basis $\mathcal{B}=\{\bm{b}_1,\dots,\bm{b}_n\}$}
\KwOutput{An obtuse basis $\mathcal{B}_o=\{\bm{b}'_1,\dots,\bm{b}'_n\}$}

Let $\mathcal{C} = \{ \bm{b}_{j_1},\bm{b}_{j_2},\dots,\bm{b}_{j_k} \}$ be a clique of mutually obtuse basis vectors\\
\While{$|\mathcal{C}|<n$}{
    Find the vector $\bm{b}_i$ in $\mathcal{B}$ which is not in $\mathcal{C}$ and is obtuse with the maximum number, $t$ of vectors in $\mathcal{C}$ \\
    \If{$k == t$} {
        Push $\bm{b}_i$ in  $\mathcal{C}$ \\
    }
    \Else{
        Calculate $\bm{A}$, the dot product matrix of $\mathcal{C}$ \\
        Calculate $\bm{c}_0 = \{-\bm{b}_i \cdot \bm{b}_{j_l} : \forall 1\leq l\leq k$ \\
        Calculate $\bm{c}_1$ = \{ $c_{0,j} - \norm{\bm{b}_{j_l}} : \forall 1\leq l\leq k$ \\
        Calculate $\bm{p}_0$ = $\bm{A}^{-1}\bm{c}_0$ and $\bm{p}_1$ = $\bm{A}^{-1}\bm{c}_1$ \\
        Sample between points $\bm{p}_0$ and $\bm{p}_1$ to find the integer point $\bm{v}$ that is the closest to $\bm{p}_0$ that satisfies $\bm{A}\bm{v} \leq \bm{c}_0$ \\
        Push $\bm{b}'_i = \bm{b}_i + \bm{v}_1 \bm{b}_{j_1} + \bm{v}_2 \bm{b}_{j_2}+ \dots +\bm{v}_k \bm{b}_{j_k}$ into $\mathcal{C}$
    }
    $k=k+1$
}
Return basis $\mathcal{C}$
\caption{Iterative Method to make Basis Obtuse}
\end{algorithm}
\end{center}

The complete algorithm is summarized in Algorithm \ref{obtuse1}, we will explain the parts of the algorithm in detail in the coming subsections.

\subsection{Initial Point of Sampling}
As discussed above, we need an integer point $\bm{v}$ such that $\bm{A}\bm{v} \leq \bm{c} $. Let us suppose $k$ to be the length of the vector $\bm{v}$. This inequality results in a geometric region inside a $k$-dimensional hyperspace, where $\bm{v}$ can lie. We will call this region the correct region in the subsequent text.

\begin{definition}
    For a set of inequalities with $k$ variables, we define the \textbf{correct region} to be the set points on $k$-dimensional hyperspace which satisfy the set of inequalities $\bm{A}\bm{v}\leq\bm{c}$.
\end{definition}

Let us call the point $\bm{A}^{-1}\bm{c}$ for $\bm{p}_0$. As $\bm{p}_0$ may not be an integer point, we ought to find an integer point which is in the correct region and as close to $\bm{p}_0$ as possible. The inequalities are linear, which means that every row of $\bm{A}$ along with the corresponding element of $\bm{c}$ represents a hyper-plane.   Each hyper-plane divides the space into two halves, where one part is where our point can be.  Hence, the complete inequality set results in $k$ planes and $2^k$ regions.  We also need the vector $\bm{v}$ to be as close to $\bm{p}_0$ as possible to keep the new basis vector as orthogonal as possible to the others.

We propose a heuristic approach to solve this ILP problem. In the geometric region depicted by the inequality $\bm{A}\bm{v} \leq \bm{c}$, we first find the line, we call it the \emph{Center Axis} in the subsequent text, such that for any point on the line, its distance from each boundary hyper-plane is equal. We do so to make sure that when we round a point on the center axis, we are certain with a high probability that the point we get is still in the correct region.  Figure \ref{fig:3dregion} shows an image of three different planes in $\mathbb{R}^3$ and the center axis of the correct region.

\begin{figure}
    \centering
    \includegraphics[width=.6\textwidth]{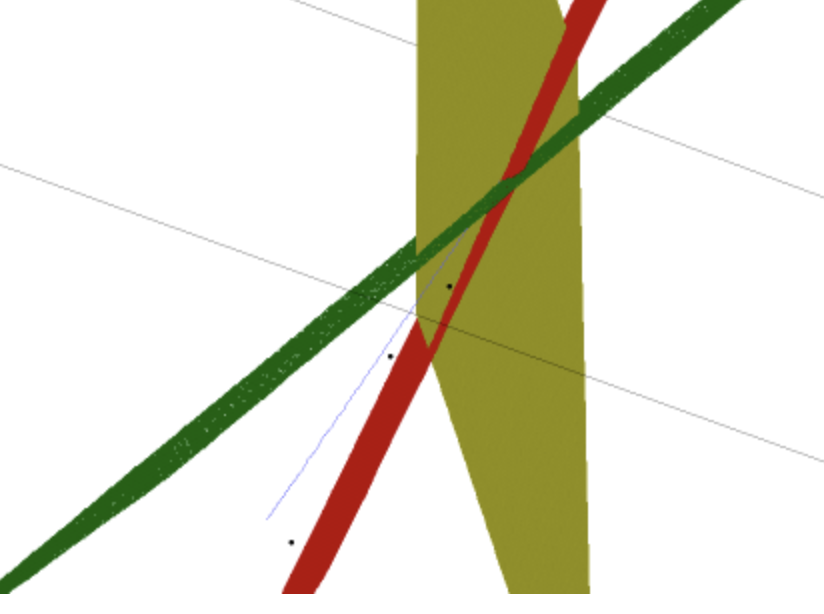}
    \caption{The blue line is the center axis of one region defined by three planes.  The black dots are integer points, rounded from sampled points on the center axis.}
    \label{fig:3dregion}
\end{figure}

To find the center axis, we have to equate the distance of any point on the line to the planes. 
\begin{lemma}
    For a hyper-plane defined by $\bm{t} \cdot \bm{x} - c = 0$ and a point $\bm{p}$, the orthogonal (or minimum) distance between the plane and the point is $\frac{\bm{t} \cdot \bm{p} - c}{\norm{\bm{t}}}$. The sign of the value depicts the side of the plane where the point lies. That is, if the signs of this value is positive for two points, then the points are on the same side of the plane.
\end{lemma}
To also incorporate the right region, we keep track of the sign of the inequality too. For a $k$-dimensional system of linear inequalities, there are $2^{k}$ center axes, but we need that one line which is in the correct region. Let $\bm{A}_i$ represent the $i$-th row of the matrix $\bm{A}$. So, for a point $\bm{x}$ on the center axis we have

\begin{equation}\label{eq:distPointPlane}
    \frac{\bm{A}_1 \cdot \bm{x} - c_1}{\norm{\bm{A}_1}} = \frac{\bm{A}_2 \cdot \bm{x} - c_2}{\norm{\bm{A}_2}} = \dots = \frac{\bm{A}_k \cdot \bm{x} - c_k}{\norm{\bm{A}_k}} = \lambda
\end{equation}

Notice that for the numerators, we have not calculated the absolute value because we want to incorporate the fact that we want to keep the line in the correct region. We will now derive the parametric form of the center axis -

\begin{equation*}
    \bm{A}_1 \cdot \bm{x} - c_1 = \lambda \norm{\bm{A}_1}
\end{equation*}
\begin{equation*}
    \bm{A}_2 \cdot \bm{x} - c_2 = \lambda \norm{\bm{A}_2}
\end{equation*}
\begin{equation*}
    \vdots    
\end{equation*}
\begin{equation*}
    \bm{A}_k \cdot \bm{x} - c_k = \lambda \norm{\bm{A}_k}
\end{equation*}
\begin{equation*}
    \bm{A}\bm{x} - \bm{c} = \lambda \begin{bmatrix}\norm{\bm{A}_1}\\\norm{\bm{A}_2}\\\vdots\\\norm{\bm{A}_k}\end{bmatrix}
\end{equation*}
\begin{equation*}
    \bm{x} = \bm{A}^{-1} (\bm{c} + \lambda  \begin{bmatrix}\norm{\bm{A}_1}\\\norm{\bm{A}_2}\\\vdots\\\norm{\bm{A}_k}\end{bmatrix})
\end{equation*}

Let us define $\bm{m}$ as
\begin{equation*}
    \bm{m} = \begin{bmatrix}m_1\\m_2\\\vdots\\m_k\end{bmatrix} =  \bm{A}^{-1}\begin{bmatrix}\norm{\bm{A}_1}\\\norm{\bm{A}_2}\\\vdots\\\norm{\bm{A}_k}\end{bmatrix}
\end{equation*}

Hence the equation becomes
\begin{equation*}
    \bm{x} = \bm{A}^{-1}\bm{c} + \lambda \bm{m}
\end{equation*}

Here, $\lambda$ is the parameter of the center axis.  From now on, we will call the point on the line with parameter $\lambda$ as $\bm{p}_{\lambda}$, and the point where all coefficients of a point $\bm{p}$ has been rounded to their nearest integer will be denoted by $\lfloor \bm{p}\rceil$. For the point to be in the correct region i.e $\bm{A}\bm{x} \leq \bm{c}$, we need $\lambda \leq 0$. Now, as we increase the absolute value of $\lambda$, our distance from the point $\bm{A}^{-1}\bm{c}$ increases, but the probability of $\lfloor \bm{p}_{\lambda}\rceil$ being inside the region after rounding increases. 

\begin{lemma}\label{lambda}
    $\lfloor\bm{p}_{\frac{-\sqrt{k}}{2}}\rceil$ is in the correct region
\end{lemma}

\begin{proof}
    On rounding a point, we change each of the point's co-ordinates by $0.5$ at the maximum, so for a $k$ dimensional point, the distance between the point and its rounding can be at most $\sqrt{0.5^2 + 0.5^2.....} = 0.5\sqrt{k}$. This distance should be less than the distance between the point and the hyper-planes bounding the region. As we see in equation (\ref{eq:distPointPlane}), the distance between $\bm{p}_{\lambda}$ and the planes is $\lambda$, so if $\lambda = 0.5\sqrt{k}$, then $\lfloor\bm{p}_{\lambda}\rceil$ will be inside the correct region. $\qed$
\end{proof}

This point is a good starting point to search for the optimal point. In the next section, we discuss our sampling techniques to find an optimal $\lambda$ such that $\norm{\bm{p}_{\lambda}-\bm{p}_0}$ is minimized while keeping $\lfloor\bm{p}_{\lambda}\rceil$ in the correct region.

\subsection{Sampling}

In the last section we found an integer point which is guaranteed to be in the correct region. The distance between this point and the $\bm{p}_0$ point is large and there are many other points closer to $\bm{p}_0$ that are also in the correct region after rounding. To find those points, we use a sampling technique. We have an upper limit of the absolute value of $\lambda$ to be $0.5\sqrt{k}$ and a lower limit of $0$. We have to choose the value between these limits and see if the rounded off point is still in the right region. The smaller the value of $\lambda$, the closer the point is to $\bm{p}_0$.

\begin{figure}
    \centering
    \includegraphics[width=.6\textwidth]{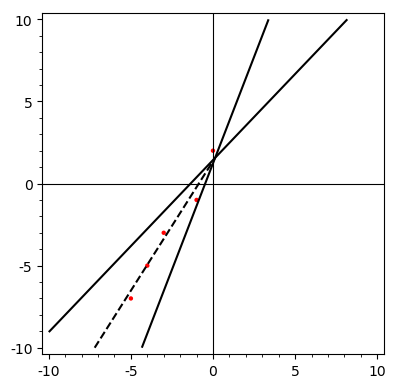}
    \caption{The dashed line represents the center axis of the correct region and the red dots represent the sampled points using $\delta = 0.5$}
    \label{fig:sampling}
\end{figure}

We have two approaches for sampling. The first approach is to determine a number $\delta$ and then starting from $0$ we decrease the value of $\lambda$ by $\delta$ to see when $\lfloor\bm{p}_{\lambda}\rceil$ is in the correct region.

\begin{center}
\begin{algorithm}[H]\label{sampling1}
\DontPrintSemicolon
\KwInput{Matrix $\bm{A}$, RHS Vector $\bm{c}$, Step-size $\delta$}
\KwOutput{Point $\bm{p}$ for which $\bm{A}\bm{p} \leq \bm{c}$ and $\norm{\bm{p} - \bm{p}_0}$ is minimized}
Let $\lambda = 0$\\
\While{$\lambda \geq -0.5\sqrt{k}$}{
    Calculate $\bm{p}_\lambda = \bm{A}^{-1}\bm{c} + \lambda \bm{m}$ \\
    Set $\bm{r} = \lfloor \bm{p}_\lambda \rceil$ \\
    \If{$\bm{A}\bm{r} \leq \bm{c}$}{
        Return $\bm{r}$
    }
    Update $\lambda = \lambda - \delta$
}
\caption{First Sampling Method to find $\lfloor\bm{p}_{\lambda}\rceil$ in the correct region}
\end{algorithm}
\end{center}

Now, the question of determining $\delta$ arises. As we see in the formula, $\lambda$ is multiplied with $\bm{m}$. So if we make sure that on changing the value of $\lambda$ by $\delta$ we change the point's coordinates by at most 1, we are not progressing too fast along the center axis.

\begin{lemma}
    Let $\epsilon$ be the maximum of the absolute values of elements in the vector $\bm{m}$.  With $\delta = 1/\epsilon$, Algorithm \ref{sampling1} will find the integer point within a distance of $\sqrt{k}$ from the integer point closest to $\bm{p}_0$.
\end{lemma}
\begin{proof}
    Without loss of generality, let us assume that $m_1$ has the greatest magnitude among all $m_i$. On decreasing $\lambda$ by $\delta$, we see a change of
    \begin{equation*}
        \Delta p_\lambda = p_\lambda - p_{\lambda-\delta} = \delta \bm{m}
    \end{equation*}
    On setting $\delta = 1/m_1$,
    \begin{equation*}
        \Delta p_\lambda = \frac{\bm{m}}{m_1} = \begin{bmatrix}1\\m_2/m_1\\\vdots\\m_k/m_1\end{bmatrix}
    \end{equation*}
    As $m_1$ is the largest number, all the elements of $\bm{m}$ are less than or equal to $1$.  So with every step we take with $\delta = 1/m_1$, the change in any co-ordinate is at most 1. After the rounding, it might be the case that the rounded off point is not the most optimal point, but might be one of the adjacent points. Hence, the distance between this point and the integer point closest to $\bm{p}_0$ would be $\sqrt{k}$.
    \qed
\end{proof}

The second method is performing a binary search of the value of $\lambda$ over the interval $(-0.5\sqrt{k}, 0)$. One assumption we make here is that if for a given $\lambda_o$, the point $\lfloor\bm{p}_{\lambda_o}\rceil$ is not in the correct region, then for all values of $\lambda$ in the interval $(\lambda_o, 0)$, the point $\lfloor\bm{p}_{\lambda}\rceil$ would also be out of the correct region.

\begin{center}
\begin{algorithm}[H]\label{sampling2}
\DontPrintSemicolon
\KwInput{Matrix $\bm{A}$, RHS Vector $\bm{c}$}
\KwOutput{A point $\bm{p}$ for which $\bm{A}\bm{p} \leq \bm{c}$ and distance between $\bm{p}$ and $\bm{p}_0$ is minimized}
Let $\lambda_L = -0.5\sqrt{k}$, $\lambda_U = 0$ \\
\While{value of $\lfloor \bm{p}_{\lambda_M} \rceil$ is unchanged for $4$ consecutive steps}{
    Let $\lambda_M = 0.5(\lambda_L + \lambda_U)$ \\
    Calculate $\bm{p}_{\lambda_M} = \bm{A}^{-1}\bm{c} + \lambda_M \bm{m}$ \\
    \If{$\bm{A} \lfloor \bm{p}_{\lambda_M} \rceil \leq \bm{c}$}
    {
        Update $\lambda_L = \lambda_M$
    }
    \Else{
        Update $\lambda_U = \lambda_M$
    }
}
Return $\bm{p}_{\lambda_L}$
\caption{Second Sampling Method to find the value of $\lambda$}
\end{algorithm}
\end{center}

The very obvious limitation of this method is that the assumption is not true for all the cases. It is not necessarily true that if a point's rounding is not in the correct region, the points with smaller values of $\lambda$ will never have a rounded off point in the correct region.  But this method is essentially faster than the first method because the number of steps is much smaller.

\section{Testing the idea in practice}\label{sec:implementation}

For testing our findings, we have implemented optimized versions of the algorithms explained in this paper. We wrote the code in both Julia and C++, but due to non-mutability of the BigFloats in Julia, we had to drop the development in Julia. The code worked fine for normal floating point numbers. In C++, we wrote the code on top of the fplll library \cite{fplll}, which is an optimized library for running common lattice algorithms on floating point lattices.

Before running Algorithm \ref{obtuse1}, we LLL-reduce the basis. Hence, we plugged our code between the LLL function call and the SVP function call. Following is the list of functions and helper functions developed by us for testing this idea

\begin{enumerate}
    \item \textbf{convert\_to\_obtuse} - A wrapper function for the complete algorithm. Takes a basis matrix as an input, and runs the functions \textbf{linsolve\_add\_to\_clique} and \textbf{sign\_flip\_add\_to\_clique} alternately until the basis is made obtuse.
    \item \textbf{linsolve\_add\_to\_clique} - Finds the vector which is not in clique but has maximum edges with the points in clique among all non-clique vectors and transforms the vector linearly so that the vector becomes obtuse to all other vectors in clique, using the approach described in Section \ref{iterativealg}.
    \item \textbf{sign\_flip\_add\_to\_clique} - Tries to find the vectors which can be included in the clique just by flipping the sign of the elements.
\end{enumerate}

Other helper functions, which do not represent the algorithm but are important to the algorithm are listed below

\begin{enumerate}
    \item \textbf{get\_maximal\_clique} - Takes in basis matrix as input and tries to make an approximate maximal clique.
    \item \textbf{generate\_dot\_product\_matrix} - Takes in basis matrix as input and calculates dot product matrix.
    \item \textbf{update\_dot\_product\_matrix} - Takes in basis matrix, dot product matrix and index and updates the dot product matrix due to changes in the basis vector at the given index.
    \item \textbf{negate\_dot\_matrix\_product} - Takes in dot product matrix and index and updates the dot product matrix due to negation of the basis vector at the given index.
    \item \textbf{norm} - Calculates norm of a vector
    \item \textbf{is\_AX\_less\_than\_b} - Takes in matrix \textbf{A}, vector \textbf{x} and \textbf{b} and returns true if $\textbf{A}\textbf{x} \leq \textbf{b}$
    \item \textbf{linsolve} - Takes in matrix \textbf{A}, vector \textbf{x} and \textbf{b} and sets $\textbf{x} = \textbf{A}^{-1}\textbf{b}$
\end{enumerate}

We tested the code on lattices of different sizes - 5, 10 and 40. We will present the results on the basis of dimensions 10 in detail below. The immediate issue we realized in this approach was that the size of the entries of the basis increased very quickly. This was not very visible in the basis of dimensions 5, but was very prominent in the 40-dimensional basis.

To show how the basis changes, we will now do some steps of the algorithm on the LLL-reduced lattice given below.

\begin{equation*}
    \mathcal{B} = \begin{bmatrix}
2&-2& -6& 10& -2& 3& 0& -4& 5& 3 \\
3& 2& 11& 2& 1& 2& -7& -3& -1& -13 \\
3& -2& -5& 7& -3& 2& 5& 16& 5& 5 \\
7& 8& 1& -2& -3& 10& 10& 5& 1& 1 \\
11& -2& 2& 12& 7& 4& 8& 0& -6& -8 \\
12& 3& 5& -5& -7& -7& -7& -7& 3& 2 \\
5& 15& 4& 6& 2& -9& 1& -3& -3& -2 \\
6& -11& -6& 0& -9& -6& 14& 2& 0& -5 \\
-1& 10& -12& -2& -2& -2& 4& 0& -15& -6 \\
6& 1& -19& -8& 18& -8& 3& -2& 12& -4 \\
\end{bmatrix}
\end{equation*}

Here, every row of the matrix $\mathcal{B}$ represents a basis vector, but we might transpose it to make the matrix fit within the margins of the paper.  We also use 0-indexing in the subsequent text to maintain consistency with the code.

On running the function \textbf{get\_maximal\_clique} on the basis, we find that the vector represented by the 0th, 1st, 3rd and 4th row of the matrix $\mathcal{B}$ form a clique. Hence, these vectors would form the initial clique. In the first iteration call of the function \textbf{linsolve\_add\_to\_clique}, we see that the vector represented by the 8th row of $\mathcal{B}$ had the maximum number of edges with the clique and hence we will add it to the clique in this step. For this vector, we find the value of $\bm{v}$, as described in Section \ref{iterativealg} to be
\begin{equation*}
    \bm{v} = \begin{bmatrix}-1&-1&-1&-1\end{bmatrix}
\end{equation*}
using the first sampling technique and the updated value of the vector is
\begin{equation*}
    \bm{b}'_8 = \begin{bmatrix}-2&0 &-16&0& 9& -13& 9& 2& -26& -5\end{bmatrix}
\end{equation*}
The clique now contains the vectors at 0th, 1st, 3rd, 4th and 8th row.  The reader can confirm that these vectors are indeed obtuse with each other i.e. dot product of any two of these vectors is less than or equal to 0.
 Following the similar procedure, we transform the remaining basis vectors and the resulting basis is given below

\begin{equation*}
    \mathcal{B}' = \begin{bmatrix}
2&3&2&7&-11&21&-908&21224&-2&184569532\\
-2&2&-30&8&2&-13&-89&-70911&0&163550688\\
-6&11&21&1&-2&24&14&-28215&-16&-60219038\\
10&2&3&-2&-12&-38&85&-38992&0&5511663\\
-2&1&14&-3&-7&31&425&-29845&9&415682694\\
3&2&-8&10&-4&2&634&45139&-13&340218771\\
0&-7&2&10&-8&32&240&-13612&9&-546575041\\
-4&-3&22&5&0&-87&-42&829&2&130456108\\
5&-1&19&1&6&-2&16&-17434&-26&-256738052\\
3&-13&15&1&8&38&-179&-22755&-5&385918402
    \end{bmatrix}^T
\end{equation*}

To show the difference between the two sampling methods, we ran the same algorithm but with the second sampling technique.  The resulting obtuse lattice, given below, contained much bigger numbers.

\begin{equation*}
    \mathcal{B}' = \begin{bmatrix}
2&3&9&7&-11&57&-52201&81940149&-2&2733250434736588\\
-2&2&-62&8&2&-19&-4677&-262905235&0&2423455627319902\\
-6&11&38&1&-2&117&-1272&-104822218&-16&-891421475978208\\
10&2&9&-2&-12&-236&7283&-145084227&0&82270071682063\\
-2&1&37&-3&-7&198&21389&-112292955&9&6157026433264479\\
3&2&-16&10&-4&90&35902&165393053&-13&5038127192105278\\
0&-7&14&10&-8&192&11367&-51424902&9&-8094917415248580\\
-4&-3&31&5&0&-680&4306&3274983&2&1932127931520721\\
5&-1&28&1&6&-46&1003&-64792422&-26&-3802176232488309\\
3&-13&30&1&8&181&-12733&-83828280&-5&5716081663738074
    \end{bmatrix}^T
\end{equation*}

As we can see, in both cases the entries in the basis vectors grow exponentially.  For higher dimensions the numbers quickly become too big to handle, so enumeration can not be used on these lattice bases.

\section{Conclusions}\label{sec:conc}
As we saw, the problem of increasing size of the entries in the modified basis vectors is a significant one. Converting a basis to an obtuse basis do have benefits, but also removes the effect of lattice reduction algorithms. Reduction techniques which preserve the obtuse property of the basis should be the next algorithms to be studied and developed for exploiting the properties of obtuse bases. Another set of algorithms to be explored are those concerning approximate obtuse lattices, that is, where a subset of basis vectors are obtuse with each other. The exponential speedup given by obtuse bases is very significant, so further study can be important.

\bibliography{mybibfile}
\bibliographystyle{splncs04}

\end{document}